\documentclass[submission,copyright,creativecommons]{eptcs}
\providecommand{\event}{GandALF 2012} 
\usepackage{amsmath,amsthm,amssymb,bussproofs,epsfig,graphicx,verbatim}
\usepackage[roscoe]{csp-cm}
\usepackage[all]{xy}

\newcommand{\sqsub}{\,\raisebox{-.5ex}
  {$\stackrel{\textstyle\sqsubset}{\scriptstyle{\sim}}$}\,}

  \newcommand{\aasg}{\,\raisebox{0.065ex}{:}{=}\,}
\newcommand{\sbr}[1]{\lbrack \! \lbrack #1 \rbrack \! \rbrack}
\newcommand{\san}[1]{\langle #1 \rangle}
\newcommand{\ibr}[1]{[ #1 \rangle}

\newcommand{\sq}{\textit{q}}
\newcommand{\sr}{\textit{run}}
\newcommand{\sok}{\textit{ok}}
\newcommand{\srd}{\textit{read}}
\newcommand{\sd}{\textit{done}}
\newcommand{\sw}{\textit{write}}
\newcommand{\sa}{\textit{abort}}

\newtheorem{proposition}{Proposition}
\newtheorem{theorem}{Theorem}

 \theoremstyle{definition}
   \newtheorem{example}{Example}

\title{Symbolic Representation of Algorithmic Game Semantics }

\author{Aleksandar S. Dimovski
\institute{Faculty of Information-Communication Tech., FON University, Skopje, 1000, MKD }
\email{aleksandar.dimovski@fon.edu.mk}
}
\def\titlerunning{Symbolic Representation of Algorithmic Game Semantics}
\def\authorrunning{A. S. Dimovski}
\begin{document}
\maketitle

\begin{abstract}
In this paper we revisit the regular-language representation of game semantics
of second-order recursion free Idealized Algol with infinite data types.
By using symbolic values instead of concrete ones  we generalize the standard notion of regular-language and
automata representations to that of corresponding symbolic representations.
In this way terms with infinite data types, such as integers, can be expressed as finite symbolic-automata
although the standard automata interpretation is infinite.
Moreover, significant reductions of the state space of game semantics models are obtained.
This enables efficient verification of terms, which is illustrated with several examples.
\end{abstract}

\section{Introduction}

Game semantics \cite{AM1,AM2,Hyl.PCF} is a technique for compositional
modelling of programming languages, which gives both sound and
complete (fully abstract) models.
Types are interpreted by \emph{games} (or arenas) between a Player,
which represents the term being modelled, and an Opponent, which represents the environment
in which the term is used.
The two participants strictly alternate to make moves, each of which is either a question (a demand for
information) or an answer (a supply of information).
Computations (executions of terms) are interpreted as \emph{plays} of a game,
while terms are expressed as \emph{strategies}, i.e.\ sets of plays, for a game.
It has been shown that game semantics model can be given certain kinds of
concrete automata-theoretic representations \cite{DL,GM,GMO},
and so it can serve as a basis for software model checking and program analysis.
However, the main limitation of model checking in general is that it can be applied only if a finite-state
model is available. This problem arises when we want to handle
terms with infinite data types.

Regular-language representation  of game semantics of second-order
recursion-free Idealized Algol with finite data types provides
algorithms for automatic verification of a range of properties, such as
observational-equivalence, approximation, and safety.
It has the disadvantage that in the presence of infinite integer
data types the obtained automata become infinite state,
i.e. regular-languages have infinite summations, thus losing their algorithmic properties.
Similarly, large finite data types are likely to make the automata infeasible.
In this paper we redefine the (standard) regular-language representation \cite{GM}
at a more abstract level so that terms with infinite data types
can be represented as finite automata, and so various program properties
can be checked over them.
The idea is to transfer attention from the standard form of automata
to what we call symbolic automata.
The representation of values constitutes the main difference between
these two formalisms.
In symbolic automata, instead of assigning concrete values to
identifiers occurring in terms, they are left as symbols.
Operations involving such identifiers will also be left as symbols.
Some of the symbols will be guarded by boolean expressions, which
indicate under which conditions these symbols can be performed.

The paper is organised as follows. The language
we consider here is introduced in Section \ref{lang}.
Symbolic representation of algorithmic game semantics is defined in Section \ref{symbol}.
Its correctness and suitability for verification of safety properties are shown in Section \ref{corr}.
In Section \ref{array} we discuss some extensions of the language, such as arrays,
and how they can be represented in the symbolic model.
A prototype tool, which implements this translation, as well as some examples
 are desribed in Section \ref{app}.


\paragraph{Related work.} By representing game semantic models as symbolic automata,
we obtain a predicate abstraction \cite{Graf,Das} based method for verification.
In \cite{BG} it was also developed a predicate abstraction from game semantics.
This was enabled by extending the models produced using game semantics such that
the state (store) is recorded explicitly in the model by using so-called stateful plays.
However, in our work we achieved predicate abstraction in a more natural way
without changing the game semantic models, and also for terms with infinite data types.

Symbolic techniques, in which data is not represented explicitly but symbolically,
have found a number of applications. For example, symbolic execution and
verification of programs \cite{S1}, symbolic program analysis \cite{S2}, and symbolic operational
semantics of process algebras \cite{S3}.

\section{The Language} \label{lang}

Idealized Algol (IA) \cite{AM1,AM2} is a well studied language
which combines call-by-name $\lambda$-calculus with
the fundamental imperative features and locally-scoped variables.
In this paper we work with its second-order recursion-free fragment
(IA$_2$ for short).

The data types $D$ are integers and booleans ($D ::= \mathsf{int} \, \mid \, \mathsf{bool}$).
The base types $B$ are expressions, commands, and variables
($B ::= \mathsf{exp} D \, \mid \, \mathsf{com} \, \mid \, \mathsf{var} D
$). We consider only first-order function types $T$
($T ::= B \mid B \rightarrow T$).

Terms are formed by the following grammar:
\begin{equation*}
\begin{array}{l}
M ::= \! x \mid v \mid \mathsf{skip} \mid M\, \mathsf{op} \, M \mid M ; \!\! M \mid
 \mathsf{if}\,M\,\mathsf{then}\,M \,\mathsf{else}\, M \mid \mathsf{while}\,M\,\mathsf{do}\,M   \\
 \ \ \mid  M := M \mid !M \mid \mathsf{new}_D \: x\!:=\!v \: \mathsf{in} \: M
  \mid \mathsf{mkvar}_D MM \mid \! \lambda x . M \mid MM
\end{array}
\end{equation*}
where $v$ ranges over constants of type $D$.
Expression constants are infinite integers and booleans.
The standard arithmetic-logic operations $\mathsf{op}$ are employed.
We have the usual imperative constructs: sequential composition,
conditional, iteration, assignment, de-referencing, and ``do nothing'' command $\mathsf{skip}$.
Block-allocated local variables are introduced by a $new$ construct,
which initializes a variable and makes it local to a given block.
The constructor $\mathsf{mkvar}$ is used for creating ``bad'' variables.
We have the standard functional constructs
for function definition and application.
\emph{Well-typed terms} are given by typing judgements of the form $\Gamma
\vdash M : T$, where $\Gamma$ is a type \emph{context} consisting of
a finite number of typed free identifiers, i.e. of the form $x_1 : T_1, \ldots, x_k : T_k$.
Typing rules of the language are given in \cite{AM1,AM2}.

The operational semantics of our language is given for terms $\Gamma \vdash M:T$,
such that all identifiers in $\Gamma$ are variables, i.e.\  $\Gamma = x_1 : \mathsf{var}D_1, \ldots, x_k :
\mathsf{var}D_k$.
It is defined by a big-step reduction
relation:
\[\Gamma \vdash M,\mathrm{s} \Longrightarrow V,\mathrm{s}'\]
where $\mathrm{s}$, $\mathrm{s}'$ represent the \emph{state} before and
after reduction. The state is a function
assigning data values to the variables in $\Gamma$.
We denote by $V$ terms in \emph{canonical form} defined by
$V ::= \, x \, \mid \, v  \, \mid \, \lambda x. M
\, \mid \, \mathsf{skip} \, \mid \, \mathsf{mkvar}_{D}MN$.
Reduction rules are standard (see \cite{AM1,AM2} for details).

Given a term $\Gamma \vdash M : \mathsf{com}$, where all identifiers in $\Gamma$ are variables,
we say that $M$ \emph{terminates} in
state $\mathrm{s}$, written $M,\mathrm{s} \Downarrow$,
if $\Gamma \vdash M,\mathrm{s} \Longrightarrow
\mathsf{skip},\mathrm{s}'$ for some state $\mathrm{s}'$.
If $M$ is a closed term then we abbreviate the relation $M, \emptyset \Downarrow$
 with $M \Downarrow$.
We say that
a term $\Gamma\vdash M:T$ is an \emph{approximate}
of a term $\Gamma\vdash N:T$, denoted by $\Gamma\vdash M \sqsub N$, if and
only if for all terms-with-hole $C[-]:\mathsf{com}$, such that
$\vdash C[M]:\mathsf{com}$ and $\vdash C[N]:\mathsf{com}$ are
well-typed closed terms of type $\mathsf{com}$,
 if $C[M]\Downarrow$  then $C[N]\Downarrow$.
If two terms approximate each other they are considered
\emph{observationally-equivalent}, denoted by
$\Gamma\vdash M \cong N$.

\section{Symbolic Game Semantics} \label{symbol}

We start by introducing a number of syntactic categories necessary
for construction of symbolic automata.
Let $Sym$ be a countable set of symbolic names, ranged over by
upper case letters X, Y, Z.
For any finite $W \subseteq Sym$, the function $new(W)$ returns
a minimal symbolic name which does not occur in $W$, and sets $W:=W \cup new(W)$.
A minimal symbolic name not in $W$ is the one which occurs earliest
in a fixed enumeration $X_1, X_2, \ldots$ of all possible symbolic names.
A set of expressions $Exp$, ranged over by $e$, is defined as follows:
\begin{equation*}
\begin{array}{l}
e ::= \! a \mid b \\
a ::= \! n \mid X^{int} \mid a \, \mathsf{op} \, a \\
b ::= \! tt \mid ff \mid X^{bool} \mid a \, = \, a \mid a \, \leq \, a \mid \neg b \mid b \land b \\
\end{array}
\end{equation*}
where $a$ ranges over arithmetic expressions ($AExp$), and $b$ over boolean
expressions ($BExp$).
We use superscripts to denote the data type of a symbolic name $X$.
We will often omit to write them, when they are clear from the context.

Let $\mathcal{A}$ be an alphabet of letters.
We define a \emph{symbolic alphabet} $\mathcal{A}^{sym}$ induced by $\mathcal{A}$ as follows:
\[
\mathcal{A}^{sym} = \mathcal{A} \cup \{ ?X, e \mid X \in Sym, e \in Exp \}
\]
The letters of the form $?X$ are called \emph{input symbols}.
They generate new symbolic names, i.e.\ $?X$ means $\mathsf{let} \, X=new(W) \, \mathsf{in} \ldots$.
We use $\alpha$ to range over $\mathcal{A}^{sym}$.
Next we define a \emph{guarded alphabet} $\mathcal{A}^{gu}$ induced by
$\mathcal{A}$ as the set of pairs of boolean conditions and symbolic letters, i.e.\ we have:
\[
\mathcal{A}^{gu} = \{ \ibr{b, \alpha} \mid b \in BExp, \alpha \in \mathcal{A}^{sym} \}
\]
A guarded letter $\ibr{b, \alpha}$ means that the symbolic letter $\alpha$
occurs only if the boolean $b$ evaluates to true, i.e.
$if \, (b=tt) \, then \, \alpha \, else \, \emptyset$.
We use $\beta$ to range over $\mathcal{A}^{gu}$.
We will often write $\alpha$ for the guarded letter $\ibr{tt, \alpha}$.
A word $\ibr{b_1, \alpha_1} \cdot \ibr{b_2, \alpha_2} \ldots \ibr{b_n, \alpha_n}$
over guarded alphabet $\mathcal{A}^{gu}$ can be represented as a pair $\ibr{b, w}$,
where $b=b_1 \land b_2 \land \ldots \land b_n$ is a boolean
and $w=\alpha_1 \cdot \alpha_2 \ldots \alpha_n$ is a word of symbolic letters.


We now show how IA$_2$ with infinite integers is interpreted by
symbolic automata, which will be denoted by extended regular expressions.
For simplicity the translation is defined for terms in $\beta$-normal form.
If a term has $\beta$-redexes, it is first reduced to $\beta$-normal form
syntactically by substitution.
In this setting, types (arenas) are represented as \emph{guarded alphabets} of
moves, plays of a game as \emph{words} over a guarded
alphabet, and strategies as \emph{symbolic automata} (regular languages) over a
guarded alphabet. The symbolic automata and regular languages, denoted by $\mathcal S(R)$
and $\mathcal L(R)$ respectively, are specified
using \emph{extended regular expressions} $R$. They are defined
inductively over finite guarded alphabets $\mathcal A^{gu}$ using the following
operations:
\begin{equation*}
\begin{array}{c}
\emptyset \quad \varepsilon \quad \beta \quad R \cdot R' \quad R^* \quad
R+R' \quad R \cap R' \\
R \mid_{\mathcal A'^{gu}} \quad R [R' / w] \quad R^{\san{\alpha}} \quad R'
\comp_{\mathcal{B}^{gu}} R \quad R \bowtie R'
\end{array}
\end{equation*}
where $R,R'$ ranges over extended regular expressions,
$\mathcal{A}^{gu},\mathcal{B}^{gu}$ over finite guarded alphabets, $\beta \in \mathcal
A^{gu}$, $\alpha \in \mathcal A^{sym}$,
$\mathcal A'^{gu} \subseteq \mathcal A^{gu}$ and $w \in \mathcal A^{gu*}$.

Constants $\emptyset$, $\varepsilon$ and $\beta$ denote the languages
$\emptyset$, $\{ \varepsilon \}$ and $\{ \beta \}$, respectively.
Concatenation $R \cdot R'$, Kleene star $R^*$, union $R+R'$ and
intersection $R \cap R'$ are the standard operations. Restriction $R
\mid_{\mathcal A'^{gu}}$ replaces all symbolic letters from $\mathcal A'^{gu}$
with $\varepsilon$ in all words of $R$, but keeps all boolean conditions.
Substitution $R [R' / w]$ is the
language of $R$ where all occurrences of the subword $w$ have been
replaced by the words of $R'$.  Given two symbols $\alpha \in \mathcal A^{sym}$,
$\beta \in \mathcal A^{gu}$, $\beta^{\san{\alpha}}$ is a new letter obtained by tagging the
latter with the former.
If a letter is tagged more than once, we write $(\beta^{\san{\alpha_1}})^{\san{\alpha_2}}=\beta^{\san{\alpha_2,\alpha_1}}$.
We define the alphabet $\mathcal{A}^{gu \san{\alpha}} = \{ \beta^{\san{\alpha}} \mid \beta \in
\mathcal A^{gu} \}$. 
Composition of regular expressions $R'$ defined over $\mathcal{A}^{gu \san{1}} + \mathcal{B}^{gu \san{2}}$
and $R$ over $\mathcal{B}^{gu \san{2}} + \mathcal{C}^{gu \san{3}}$ is given as
follows:
\begin{equation*}
\begin{array}{l}
R' \comp_{\mathcal{B}^{gu \san{2}}} R = \{ w \big[ \ibr{b \land b_1 \land b_2 \land b'_1 \land b'_2 \land \alpha_1=\alpha'_1 \land \alpha_2=\alpha'_2,s \,}^{\san{1}} / \\
\qquad \qquad \ibr{ b_1, \alpha_1 }^{\san{2}} \cdot \ibr{b_2, \alpha_2 }^{\san{2}} \big] \ \mid  \ w \in R,
 \ibr{ b'_1, \alpha'_1 }^{\san{2}} \cdot \ibr{b,s}^{\san{1}} \cdot \ibr{b'_2, \alpha'_2 }^{\san{2}} \in R' \}
\end{array}
\end{equation*}
where $R'$ is a set of words of form $\ibr{ b'_1, \alpha'_1 }^{\san{2}} \cdot \ibr{b,s}^{\san{1}} \cdot \ibr{b'_2, \alpha'_2 }^{\san{2}}$,
such that $\ibr{b'_1, \alpha'_1 }^{\san{2}}$, $\ibr{ b'_2, \alpha'_2 }^{\san{2}} \in \mathcal{B}^{gu \san{2}}$ and $\ibr{b,s}$ contains only letters
from $\mathcal A^{gu \san{1}}$. 
So all letters of $\mathcal{B}^{gu \san{2}}$ are removed from the composition,
which is defined over the alphabet $\mathcal{A}^{gu \san{1}} + \mathcal{C}^{gu \san{3}}$.
The shuffle operation of two regular
languages is defined as $\mathcal L(R) \bowtie \mathcal L(R') =
\bigcup_{w_1 \in \mathcal L(R),w_2 \in \mathcal L(R')} w_1 \bowtie
w_2 $, where $w \bowtie \varepsilon = \varepsilon \bowtie w = w$ and
$a \cdot w_1 \bowtie b \cdot w_2 = a \cdot (w_1 \bowtie b
\cdot w_2) + b \cdot (a \cdot w_1 \bowtie w_2)$.
It is a standard result that any extended
regular expression obtained from the operations above denotes a
regular language \cite[pp.\ 11--12]{GM}, which can be recognised
by a finite (symbolic) automaton \cite{Automata}.

Each type $T$ is interpreted by a guarded alphabet of moves $\mathcal A_{\sbr{T}}^{gu}$
induced by $\mathcal A_{\sbr{T}}$.
The alphabet $\mathcal A_{\sbr{T}}$ contains two kinds of moves:
\emph{questions} and \emph{answers}.
They are defined as follows.
\begin{align*}
& \mathcal{A}_{\sbr{\mathsf{int}}} = \{ \ldots, -n,-n+1, \ldots, n,n+1, \ldots \} \qquad  \mathcal{A}_{\sbr{\mathsf{bool}}} = \{ tt, ff \} \\
& \mathcal{A}_{\sbr{\mathsf{exp}D}} = \{ \sq \} \cup \mathcal{A}_{\sbr{D}} \qquad \mathcal{A}_{\sbr{\mathsf{com}}} = \{ \sr, \sd \}  \\
& \mathcal{A}_{\sbr{\mathsf{var}D}} = \{ \srd, \sw(a), a, \sok \, \mid \, a
\in \mathcal{A}_{\sbr{D}} \} \\
& \mathcal{A}_{\sbr{B_1^{\san{1}} \to \ldots \to B_k^{\san{k}} \to B}}^{gu} = \displaystyle{\sum_{1 \leq i \leq k}} \mathcal{A}_{\sbr{B_i}}^{gu \, \san{i}} + \mathcal{A}_{\sbr{B}}^{gu}
\end{align*}
Note that function types are tagged by a superscript ($\san{i}$) in order to keep
record from which type, i.e.\ which component of the disjoint union, each move comes from.
The letters in the alphabet $\mathcal A_{\sbr{T}}$ represent
\emph{moves} (observable actions) that a term of type $T$ can perform.
For example, in $\mathcal{A}_{\sbr{\mathsf{exp}D}}$ there is a question move
$\sq$ to ask for the value of the expression, and values from $\mathcal{A}_{\sbr{D}}$ to answer the question.
For commands, in $\mathcal{A}_{\sbr{\mathsf{com}}}$ there is a question move
$\sr$ to initiate a command, and an answer move $\sd$ to signal
successful termination of a command.
For variables, we have moves for writing to the variable, $\sw(a)$,
acknowledged by the move $\sok$, and for reading from the variable, a question move $\srd$,
and corresponding to it an answer from $\mathcal{A}_{\sbr{D}}$.

For any ($\beta$-normal) term, we define
a regular-language which represents its game semantics,
i.e.\ its set of complete plays.
Every complete play represents the observable effects
of a completed computation of the given term.
It is given as a guarded word $\ibr{b,w}$, where
the boolean $b$ is also called \emph{play condition}.
Assumptions about a play (computation) to be feasible are
recorded in the play condition.
For infeasible plays, the play condition is inconsistent (unsatisfiable),
thus no assignment of concrete values to symbolic names exists that makes
the play condition true.
So it is desirable for any play to check the consistency (satisfiability) of its play condition.
If the play condition is found to be inconsistent, this play is discarded from the
final model of the corresponding term.
The regular expression  for $\Gamma \vdash M : T$ is denoted $\sbr{\Gamma \vdash M : T}$,
and it is defined over the guarded alphabet $\mathcal A_{\sbr{\Gamma \vdash T}}^{gu}$
defined as:
\begin{equation*}
\mathcal A_{\sbr{\Gamma \vdash T}}^{gu} = \big( \sum_{x : T' \in
\Gamma} \mathcal{A}_{\sbr{T'}}^{gu \, \san{x}} \big) +
\mathcal{A}_{\sbr{T}}^{gu}
\end{equation*}
Free identifiers $x:T \in \Gamma$ are represented by the
copy-cat regular expressions given in Table~\ref{free.id},
which contain all possible behaviours of terms of that type.
They provide a generic closure of an open program term.
For example, $x:\mathsf{exp}D^{\san{x}} \vdash x:\mathsf{exp}D$ is modelled
by the word $\sq \cdot \sq^{\san{x}} \cdot ?X^{\san{x}} \cdot X$.
Its meaning is that Opponent starts the play by asking what is the value
of this expression with the move $\sq$, and Player responds by playing $\sq^{\san{x}}$
(i.e.\ what is the value of the non-local expression $x$). Then Opponent provides
the value of $x$ by using a new symbolic name $X$, which will be also the value
of this expression.
Languages $R^{\san{x,i}}_{B}$ contain plays representing a function which
evaluates its $i$-th argument.

\begin{table}
\fbox{
\begin{minipage}{89ex}
$ \begin{array}{@{}l}
 \sbr{\Gamma, x:B_1^{\san{x,1}} \to \ldots B_k^{\san{x,k}} \to \mathsf{exp}D^{\san{x}} \vdash x : B_1^{\san{1}} \to \ldots B_k^{\san{k}} \to \mathsf{exp}D} =
  \sq \cdot \sq^{\san{x}} \cdot \big( \sum_{1 \leq i \leq k} R^{\san{x,i}}_{B_i} \big)^* \cdot ?X^{\san{x}} \cdot X \\
\sbr{\Gamma, x:B_1^{\san{x,1}} \to \ldots B_k^{\san{x,k}} \to \mathsf{com}^{\san{x}} \vdash x : B_1^{\san{1}} \to \ldots B_k^{\san{k}} \to \mathsf{com}} = \\
 \qquad \qquad \qquad \qquad \qquad \qquad \qquad \qquad \qquad \qquad \sr \cdot \sr^{\san{x}} \cdot \big( \sum_{1 \leq i \leq k} R^{\san{x,i}}_{B_i} \big)^* \cdot \sd^{\san{x}} \cdot \sd \\
  \sbr{\Gamma, x:B_1^{\san{x,1}} \to \ldots B_k^{\san{x,k}} \to \mathsf{var}D^{\san{x}} \vdash x : B_1^{\san{1}} \to \ldots B_k^{\san{k}} \to \mathsf{var}D} = \\
 \qquad  \big( \srd \cdot \srd^{\san{x}} \cdot \big( \sum_{1 \leq i \leq k} R^{\san{x,i}}_{B_i} \big)^* \cdot ?Z^{\san{x}} \cdot Z \big) + \big( \sw(?Z') \cdot \sw(Z')^{\san{x}} \cdot \big( \sum_{1 \leq i \leq k} R^{\san{x,i}}_{B_i} \big)^* \cdot \sok^{\san{x}} \cdot \sok \big) \\
 R^{\san{x,i}}_{\mathsf{exp}D} = \sq^{\san{x,i}} \cdot \sq^{\san{i}} \cdot ?Z^{\san{i}} \cdot Z^{\san{x,i}} \\
  R^{\san{x,i}}_{\mathsf{com}} = \sr^{\san{x,i}} \cdot \sr^{\san{i}} \cdot \sd^{\san{i}} \cdot \sd^{\san{x,i}} \\
  R^{\san{x,i}}_{\mathsf{var}D} = \! ( \srd^{\san{x,i}} \cdot \srd^{\san{i}} \cdot ?Z^{\san{i}} \cdot Z^{\san{x,i}} ) \! +
\!  ( \sw(?Z')^{\san{x,i}} \cdot \sw(Z')^{\san{i}} \cdot \sok^{\san{i}} \cdot \sok^{\san{x,i}} )\\
 \end{array} $
\end{minipage}
} \caption{Free Identifiers} \label{free.id}
\end{table}

Note that whenever an input symbol $?X$ is met in a play,
a new symbolic name is created, which binds all occurrences of $X$
that follow in the play until a new $?X$ is met.
For example,
$\sbr{f:\mathsf{expint}^{\san{f,1}} \to \mathsf{expint}^{\san{f}} \vdash f : \mathsf{expint}^{\san{1}} \to \mathsf{expint}} =  \sq \cdot \sq^{\san{f}} \cdot \big( \sq^{\san{f,1}} \cdot \sq^{\san{1}} \cdot ?Z^{\san{1}} \cdot Z^{\san{f,1}} \big)^* \cdot ?X^{\san{f}} \cdot X$
is a model for a non-local function $f$ which may evaluate its argument zero or more times.
The play corresponding to $f$ which evaluates its argument two times
is given as: $\sq \cdot \sq^{\san{f}} \cdot \sq^{\san{f,1}} \cdot \sq^{\san{1}} \cdot Z_1^{\san{1}} \cdot Z_1^{\san{f,1}} \cdot
\sq^{\san{f,1}} \cdot \sq^{\san{1}} \cdot Z_2^{\san{1}} \cdot Z_2^{\san{f,1}} \cdot X^{\san{f}} \cdot X$.
Note that letters tagged with $f$ represent the actions of calling and
returning from the function, while letters tagged with $f.1$ are the actions caused
by evaluating the first argument of $f$.

In Table~\ref{csp.rl1} terms are interpreted by regular
expressions describing their sets of complete plays.
An integer or boolean constant is modeled by a play where the
initial question $\sq$ is answered by the value of that constant.
The only play for $\mathsf{skip}$ responds to $\sr$ with $\sd$.
A composite term $\mathsf{c}(M_1,\ldots,M_k)$ consisting of a language
construct `$\mathsf{c}$' and subterms $M_1,\ldots,M_k$ is interpreted by
composing the regular expressions for $M_1,\ldots,M_k$, and a regular expression for `$\mathsf{c}$'.
The representation of language constructs `$\mathsf{c}$' is given in Table~\ref{csp.rl2}.
In the definition for local variables, a `cell' regular expression $\gamma_{v}^{x}$
is used to remember the initial and the most-recently written value into the variable $x$.
Notice that all symbols used in Tables~\ref{free.id},\ref{csp.rl1},\ref{csp.rl2} are
of data type $D$, except the symbol $Z$ in $\mathsf{if}$ and $\mathsf{while}$
constructs, which is of data type $bool$.

\begin{table}
\fbox{
\begin{minipage}{89ex}
$ \begin{array}{@{}l}
 \sbr{\Gamma \vdash v : \mathsf{exp}D} = \sq \cdot v \\
 \sbr{\Gamma \vdash \mathsf{skip} : \mathsf{com}} = \sr \cdot \sd \\
 \sbr{\Gamma \vdash \mathsf{c}(M_1,\ldots,M_k) : B'} = \sbr{\Gamma \vdash M_1 : B_1^{\san{1}}} \comp_{\mathcal{A}_{\sbr{B_1}}^{gu \, \san{1}}}  \cdots \\
 \qquad \qquad \qquad \qquad \qquad \qquad \cdots \ \sbr{\Gamma \vdash M_k : B_k^{\san{k}}} \comp_{\mathcal{A}_{\sbr{B_k}}^{gu \, \san{k}}} \sbr{\mathsf{c}:B_1^{\san{1}} \times \ldots B_k^{\san{k}} \to B'}  \\
 \sbr{\Gamma \vdash M N : T} = \sbr{\Gamma \vdash N : B^{\san{1}}} \comp_{\mathcal{A}_{\sbr{B}}^{gu \, \san{1}}} \sbr{\Gamma \vdash M : B^{\san{1}} \to T} \\
 \sbr{\Gamma \vdash \mathsf{new}_D \, x:=v \, \mathsf{in} \, M:B} = \big( \sbr{\Gamma,x\!:\!\mathsf{var}D \vdash M} \cap ( \gamma_{v}^{x} \bowtie \mathcal{A}_{\sbr{\Gamma \vdash B}^{gu \, *}}  ) \big) \!\! \mid_{\mathcal{A}_{\sbr{\mathsf{var}D}}^{\san{x}}} \\
 \qquad \gamma_{v}^{x} = (\srd^{\san{x}} \cdot v^{\san{x}})^* \cdot \big( \sw(?Z)^{\san{x}}
 \cdot \sok^{\san{x}} \cdot (\srd^{\san{x}} \cdot Z^{\san{x}})^*
 \big)^* \\
\end{array} $
\end{minipage}
} \caption{Language terms} \label{csp.rl1}
\end{table}

\begin{table}[h]
\fbox{
\begin{minipage}{89ex}
$ \begin{array}{l}
 \sbr{\mathsf{op} : \mathsf{exp}D_1^{\san{1}} \times \mathsf{exp}D_2^{\san{2}} \to \mathsf{exp}D } =
 \sq \cdot \sq^{\san{1}} \cdot ?Z^{\san{1}} \cdot \sq^{\san{2}} \cdot ?Z'^{\san{2}} \cdot (Z \, \mathsf{op} \, Z') \\

 \sbr{\mathsf{;} : \mathsf{com}^{\san{1}} \times \mathsf{com}^{\san{2}} \to \mathsf{com} } =
  \sr \cdot \sr^{\san{1}} \cdot \sd^{\san{1}} \cdot \sr^{\san{2}} \cdot \sd^{\san{2}} \cdot \sd \\

 \sbr{\mathsf{if} : \mathsf{expbool}^{\san{1}} \times \mathsf{com}_1^{\san{2}} \times \mathsf{com}_2^{\san{3}} \to \mathsf{com} } =
  \ibr{tt, \sr }  \cdot \ibr{tt, \sq^{\san{1}} }  \cdot \ibr{tt, ?Z^{\san{1}} } \cdot \\
 \qquad \qquad \qquad \qquad \qquad \qquad \qquad \big( \ibr{Z, \sr^{\san{2}} } \cdot \ibr{tt, \sd^{\san{2}} } + \ibr{\neg Z, \sr^{\san{3}} } \cdot \ibr{tt, \sd^{\san{3}} } \big) \cdot \ibr{tt, \sd } \\

 \sbr{\mathsf{while} : \mathsf{expbool}^{\san{1}} \times \mathsf{com}^{\san{2}} \to \mathsf{com} } =
   \ibr{tt, \sr }  \cdot \ibr{tt, \sq^{\san{1}} }  \cdot \ibr{tt, ?Z^{\san{1}} } \cdot \\
  \qquad \qquad \qquad \qquad \qquad \qquad \qquad \big( \ibr{Z, \sr^{\san{2}} } \cdot \ibr{tt, \sd^{\san{2}} } \cdot \ibr{tt, \sq^{\san{1}} }  \cdot \ibr{tt, ?Z^{\san{1}} } \big)^{*} \cdot \ibr{\neg Z, \sd}  \\

 \sbr{\mathsf{:=} : \mathsf{var}D^{\san{1}} \times \mathsf{exp}D^{\san{2}} \to \mathsf{com} } =
   \sr \cdot \sq^{\san{2}} \cdot ?Z^{\san{2}} \cdot \sw(Z)^{\san{1}} \cdot \sok^{\san{1}} \cdot \sd \\

 \sbr{\mathsf{!} : \mathsf{var}D^{\san{1}} \to \mathsf{exp}D } =
  \sq \cdot \srd^{\san{1}} \cdot ?Z^{\san{1}} \cdot Z  \\
\end{array} $
\end{minipage}
} \caption{Language constructs} \label{csp.rl2}
\end{table}

We define an effective alphabet of a regular expression to be the set
of all letters appearing in the language denoted by that regular expression.
Then we can show.

\begin{proposition}
For any term $\Gamma \vdash M:T$, the effective alphabet of $\sbr{\Gamma \vdash M:T}$
is a finite subset of $\mathcal{A}_{\sbr{\Gamma \vdash T}}^{gu}$.
\end{proposition}

Any term $\Gamma \vdash M : T$ from IA$_2$ with infinite integers is
interpreted by extended regular expression without infinite summations
defined over finite alphabet. So the following is immediate.

\begin{theorem} \label{sym_aut}
For any IA$_2$ term, the set $\mathcal{L} \sbr{\Gamma \vdash M : T}$
is a symbolic regular-language without infinite summations over finite alphabet.
Moreover, a finite symbolic automata $\mathcal S \sbr{\Gamma \vdash M : T}$
which recognizes it is effectively constructible.
\end{theorem}
\begin{proof}
The proof is by induction on the structure of $\Gamma \vdash M : T$.

An automaton is a tuple $(Q,i,\delta,F)$ where $Q$ is the finite set
of states, $i \in Q$ is the initial state, $\delta$ is the transition
function, and $F \subseteq Q$ is the set of final states.
We now introduce two auxiliary operations.
Let $A'=(Q',i',\delta',F')$ be an automaton, then
$A=rename(A',tag)$ is defined as:\\
$\begin{array}{l}
 Q=Q' \qquad i=i' \qquad F=F' \\
 \delta = \{ q_1 \stackrel{\ibr{b,m}}{\longrightarrow} q_2 \in \delta' \mid q_1 \neq i', q_2 \nin F' \} \, +  \\
 \qquad  \{ i' \stackrel{\ibr{b,m^{\san{tag}}}}{\longrightarrow} q  \mid i' \stackrel{\ibr{b,m}}{\longrightarrow} q \in \delta' \} \, +
\{ q_1 \stackrel{\ibr{b,m^{\san{tag}}}}{\longrightarrow} q_2  \mid q_1 \stackrel{\ibr{b,m}}{\longrightarrow} q_2 \in \delta', q_2 \in F' \}
\end{array}$

Let $A_1=(Q_1,i_1,\delta_1,F_1)$ and $A_2=(Q_2,i_2,\delta_2,F_2)$ be two automata,
such that all transitions going out of $i_2$ and going to a state from $F_2$
are tagged with $tag$. Define $A=compose(A_1,A_2,tag)$ as follows:\\
$\begin{array}{l}
 Q=Q_1+Q_2 \backslash \{ i_2, F_2 \} \qquad i=i_1 \qquad F=F_1 \\
 \delta = \{ q_1 \stackrel{\ibr{b,m}}{\longrightarrow} q'_1 \in \delta_1 \mid m \neq n^{\san{tag}} \} \, +  \,
 \{ q_2 \stackrel{\ibr{b,m}}{\longrightarrow} q'_2 \in \delta_2 \mid m \neq n^{\san{tag}} \} \, +  \\
 \qquad \{ q_1 \stackrel{\ibr{b_1 \land b_2 \land m_1=m_2,\varepsilon}}{\longrightarrow} q'_2  \mid q_1 \stackrel{\ibr{b_1,m^{\san{tag}}_{1}}}{\longrightarrow} q'_1 \in \delta_1, i_2 \stackrel{\ibr{b_2,m^{\san{tag}}_{2}}}{\longrightarrow} q'_2 \in \delta_2, \{m_1, m_2\} \, \textrm{are} \ \textrm{questions}  \} \, + \\
 \qquad \{ q_2 \stackrel{\ibr{b_1 \land b_2 \land m_1=m_2,\varepsilon}}{\longrightarrow} q'_1  \mid q_1 \stackrel{\ibr{b_1,m^{\san{tag}}_{1}}}{\longrightarrow} q'_1 \in \delta_1, q_2 \stackrel{\ibr{b_2,m^{\san{tag}}_{2}}}{\longrightarrow} q'_2 \in \delta_2,
 \ q'_2 \in F_2, \{m_1, m_2\} \, \textrm{are} \ \textrm{answers}  \} \\
\end{array}$

Let $A_M$, $A_N$, and $A_{\comp}$ be automata representing $\Gamma \vdash M$,
$\Gamma \vdash N$, and construct $;$ (see Table~\ref{csp.rl2}), respectively. The unique automaton representing
$\Gamma \vdash M \, ; N$ is defined as:
\[
A_{M \, ; N} = compose(compose(A_{\comp},rename(A_M,1),1),rename(A_N,2),2)
\]
The other cases for constructs are similar.

The automaton $A=(Q,i,\delta,F)$ for $\sbr{\Gamma \vdash \mathsf{new}_D \, x:=v \, \mathsf{in} \, M}$
is constructed in two stages. First we eliminate $x$-tagged symbolic letters from
$A_M=(Q_M,i_M,\delta_M,F_M)$, which represents $\sbr{\Gamma, x:\mathsf{var}D \vdash M}$, by replacing
them with $\varepsilon$.
We introduce a new symbolic name $X$ to keep track of what changes to $x$ are made
by each $x$-tagged move.\\
$\begin{array}{l}
 Q_{\varepsilon}=Q_M \qquad i_{\varepsilon}=i_M \qquad F_{\varepsilon}=F_M \\
 \delta_{\varepsilon} \! = \! \{ i_M \stackrel{\ibr{?X=v \land b,m}}{\longrightarrow} q  \mid i_M \stackrel{\ibr{b,m}}{\longrightarrow} q \in \delta_M \} \ + \\
 \qquad \{ q_1 \stackrel{\ibr{b,m}}{\longrightarrow} q_2  \mid q_1 \stackrel{\ibr{b,m}}{\longrightarrow} q_2 \in \delta_M, m \nin \!\! \{ \sw(a)^{\san{x}}, \sok^{\san{x}}, \srd^{\san{x}}, a^{\san{x}} \! \} \}  \\
 \qquad \{ q_1 \stackrel{\ibr{?X=a' \land b_1 \land b_2,\varepsilon}}{\longrightarrow} q_2  \mid \exists q. ( q_1 \stackrel{\ibr{b_1,\sw(a')^{\san{x}}}}{\longrightarrow} q \in \delta_M, q \stackrel{\ibr{b_2,\sok^{\san{x}}}}{\longrightarrow} q_2 \in \delta_M )  \} \\
\qquad \{ q_1 \stackrel{\ibr{a'=X \land b_1 \land b_2,\varepsilon}}{\longrightarrow} q_2  \mid \exists q. ( q_1 \stackrel{\ibr{b_1,\srd^{\san{x}}}}{\longrightarrow} q \in \delta_M, q \stackrel{\ibr{b_2,{a'}^{\san{x}}}}{\longrightarrow} q_2 \in \delta_M )  \} \\
\end{array}$

The final automaton is obtained by eliminating $\varepsilon$-letters from $A_{\varepsilon}$.
Note that conditions associated to $\varepsilon$-letters are not removed.\\
$\begin{array}{l}
 Q=Q_{\varepsilon} \qquad i=i_{\varepsilon} \qquad F=F_{\varepsilon} \\
 \delta =  \big( \{ \delta_{\varepsilon} \backslash \{ q_1 \stackrel{\ibr{b,\varepsilon}}{\longrightarrow} q_2  \mid q_1, q_2 \in Q_{\varepsilon} \} \big) \, +  \,  \{ q_1 \stackrel{\ibr{b \land b_{\varepsilon},m}}{\longrightarrow} q_2 \mid \exists q' \in Q_{\varepsilon}. ( q_1 \stackrel{\ibr{b_{\varepsilon},\varepsilon}^*}{\longrightarrow} q', q' \stackrel{\ibr{b,m}}{\longrightarrow} q_2 )  \} \\
\end{array}$\\
We write $q_1 \stackrel{\ibr{b_{\varepsilon},\varepsilon}^*}{\longrightarrow} q_2$ if $q_2$
is reachable from $q_1$ by a series of $\varepsilon$-transitions $\ibr{b_1,\varepsilon}, \ldots, \ibr{b_k,\varepsilon}$,
where $b_{\varepsilon} = b_1 \land \ldots b_k$.
\end{proof}



\begin{example}
Consider the term $M_1$:
\[
\begin{array}{l}
f : \mathsf{com}^{f,1} \to \mathsf{com}^{f}, abort:\mathsf{com}^{abort}, x: \mathsf{expint}^{x}, y: \mathsf{expint}^{y}  \vdash
 f \big( \mathsf{if} \,  (x \neq y) \, \mathsf{then} \, abort \big) : \mathsf{com}
\end{array}
\]
in which $f$ is a  non-local procedure, and $x$, $y$ are non-local expressions.

\begin{figure}
 \centerline{\scalebox{1.6}{\psfig{figure=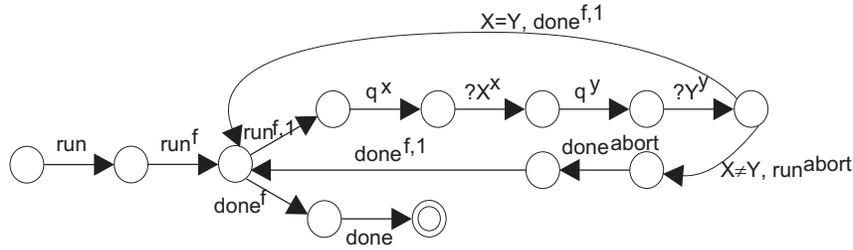}}}
 \caption{The symbolic representation of the strategy for $M_1$. }
\label{warm}
\end{figure}
The strategy for this term represented as a finite symbolic automaton is
shown in Figure~\ref{warm}.
The model illustrates only the possible behaviors of this term: the
non-local procedure $f$ may call its argument, zero or more times,
then the term terminates successfully with $\sd$.
If $f$ calls its argument, arbitrary values for
$x$ and $y$ are read from the environment by using symbols $X$ and $Y$. If they are different ($X \neq Y$),
then the $\mathsf{abort}$ command is executed.
The standard regular-language representation \cite{GM} of $M_1$,
where concrete values are employed, is given in Figure~\ref{standard}.
It represents an infinite-state automaton, and so it is not suitable
for automatic verification (model checking).
Note that, the values for non-local expressions $x$ and $y$ can be any possible integer.
\begin{figure}
 \centerline{\scalebox{1.6}{\psfig{figure=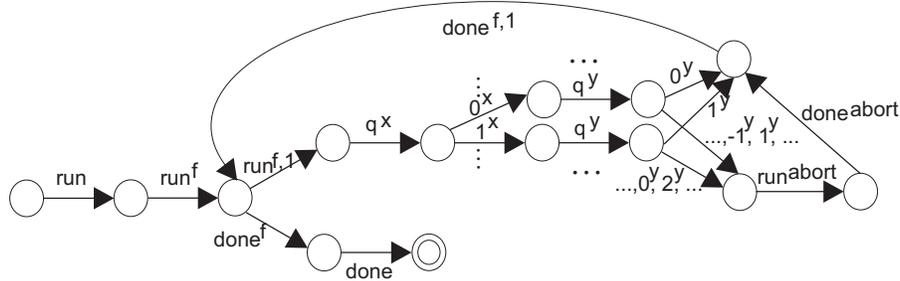}}}
 \caption{The standard representation of the strategy for $M_1$. }
\label{standard}
\end{figure}
$\Box$
\end{example}

\section{Formal Properties} \label{corr}

In \cite[pp.\ 28--32]{GM}, it was shown the correctness of
the standard regular-language representation for finitary IA$_2$  by showing that
it is isomorphic to the game semantics model \cite{AM1}.
As a corollary, it was obtained that the standard regular-language representation
is fully abstract.

Let $\sbr{\Gamma \vdash M:T}^{CR}$ denotes the set of all complete plays
in the strategy for a term $\Gamma \vdash M:T$ from $IA_2$ with infinite integers
obtained as in \cite{GM}, where concrete values in moves and infinite summations
in regular expressions are used.
Suppose that there is a special free identifier $\mathsf{abort}$ of type $\mathsf{com}$.
A term is $\mathsf{abort}$-free if it has no occurrence of $\mathsf{abort}$.
We say that a term is \emph{safe} if for any $\mathsf{abort}$-free term-with-hole $C[-]$,
the term $C[M]$ does not execute the $\mathsf{abort}$ command.
Since the standard regular-language semantics is fully abstract, the following
result is easy to show.
\begin{proposition} \label{safe}
A term $M$ is safe
if $\sbr{\Gamma\vdash M}^{CR}$ does not contain moves from $\mathcal{A}_{\sbr{\mathsf{com}}}^{\mathsf{abort}}$.
\end{proposition}

Let $Eval$ be the set of evaluations, i.e.\ the set of total functions from
$W$ to $\mathcal{A}_{\sbr{\mathsf{int}}} \cup \mathcal{A}_{\sbr{\mathsf{bool}}}$.
We use $\rho$ to range over $Eval$.
So we have $\rho(X^D) \in \mathcal{A}_{\sbr{D}}$ for any evaluation $\rho$
and $X^D \in W$.
Given a word of symbolic letters $w$, let $\rho(w)$ be a word where every
symbolic name is replaced by the corresponding concrete value as defined by $\rho$.
Given a guarded word $\ibr{b,w}$, define $\rho(\ibr{b,w})= \rho(w)$ if $\rho(b)=tt$;
otherwise $\rho(\ibr{b,w})= \emptyset$ if $\rho(b)=ff$.
The concretization of a symbolic regular-language over a guarded alphabet is
defined as follows:
$\gamma \, \mathcal{L} (R) = \{ \rho \ibr{b,w} \mid \ibr{b,w} \in \mathcal{L} (R), \rho \in Eval  \}$.
Let $\sbr{\Gamma \vdash M:T}^{SR} = \mathcal{L} \sbr{\Gamma \vdash M:T}$ be the strategy obtained as in Section 3, where
symbols instead of concrete values are used.

\begin{theorem} \label{rl.srl}
For any IA$_2$ term
\begin{equation*}
{\gamma \, \sbr{\Gamma \vdash M:T}^{SR}} \ = \ {\sbr{\Gamma \vdash M:T}}^{CR}
\end{equation*}
\end{theorem}
\begin{proof}
By induction on the typing rules.
The definitions of expression and command constructs are the same.

Consider the case of free identifiers.\\
$\begin{array}{l}
 \gamma \sbr{x:\mathsf{exp}D^{\san{x}} \vdash x:\mathsf{exp}D}^{SR} = \gamma \{ \sq \cdot \sq^{\san{x}} \cdot X^{D \, \san{x}} \cdot X^D  \} \\
\qquad \qquad \qquad \qquad = \{ \sq \cdot \sq^{\san{x}} \cdot \rho(X^D)^{\san{x}} \cdot \rho({X^D}) \mid \rho : \{ X^D \} \to \mathcal{A}_{\sbr{D}} \} \\
\qquad \qquad \qquad \qquad = \{  \sq \cdot \sq^{\san{x}} \cdot v^{\san{x}} \cdot v \mid v \in \mathcal{A}_{\sbr{D}}  \} = \sbr{x:\mathsf{exp}D^{\san{x}} \vdash x:\mathsf{exp}D}^{CR}
\end{array}$\\
The other cases are similar to prove.
\end{proof}
As a corollary we obtain the following result.
\begin{theorem} \label{saf}
${\sbr{\Gamma \vdash M:T}^{SR}}$ is safe iff ${\sbr{\Gamma \vdash N:T}^{CR}}$ is safe.
\end{theorem}

By Proposition~\ref{safe} and Theorem~\ref{saf} it follows that a term is safe if its symbolic regular-language
semantics is safe.
Since symbolic automata are finite state, it follows that we can use model-checking to verify  safety of
IA$_2$ terms with infinite data types.

In order to verify safety of a term we need to check whether the symbolic automaton representing a term contains
unsafe plays.
We use an external SMT solver Yices \footnote{http://yices.csl.sri.com} \cite{Yices} to determine consistency of the play conditions
of the discovered unsafe plays. If some play condition is consistent,
i.e.\ there exists an evaluation $\rho$ that makes the play condition true,
the corresponding unsafe play is feasible and it is reported as a genuine counter-example.

\begin{example}
The term $M_1$ from Example 1 is $\mathsf{abort}$-unsafe, with the following counter-example:
\[
\sr \ \sr^{f} \  \sr^{f,1} \ \ \sq^{x} \ \ X^{x} \ \ \sq^{y} \ \ Y^{y} \ \ \ibr{X \neq Y,\sr^{abort}} \ \ \sd^{abort} \ \sd^{f,1} \ \sd^{f} \ \sd \]
The consistency of the play condition is established by instructing Yices to
check the formula:
\[
\begin{array} {l}
(define \, X::int) \\
(define \, Y::int) \\
(assert \, (/= \, X \ Y))
\end{array}
\]
The following satisfiable assignments to symbols are reported: $X=1$ and $Y=2$,
yielding a concrete unsafe play:
$\sr \ \sr^{f} \  \sr^{f,1} \ \ \sq^{x} \ \ 1^{x} \ \ \sq^{y} \ \ 2^{y} \ \ \sr^{abort} \ \ \sd^{abort} \ \sd^{f,1} \ \sd^{f} \ \sd$.
$\Box$
\end{example}

\begin{example}
Consider the term $M_2$:
\[
\begin{array}{l}
N : \mathsf{expint}^{N}, abort:\mathsf{com}^{abort} \vdash \, \mathsf{new_{int}} \, x:=0 \ \mathsf{in} \\
\qquad \qquad \qquad \qquad \qquad \qquad \ \mathsf{while} \, (x<N) \ \mathsf{do} \ x:=x+1; \\
\qquad \qquad \qquad \qquad \qquad \qquad \ \mathsf{if} \, (x>0) \ \mathsf{then} \ \mathsf{abort} : \mathsf{com}
\end{array}
\]

The strategy for this term (suitably adapted for readability) is given in Figure~\ref{warm2}.
Observe that the term communicates with its environment using non-local
identifiers $N$ and $\mathsf{abort}$. So in the model will only be represented
actions of $N$ and $\mathsf{abort}$.
Notice that each time the term (Player) asks for a value of $N$ with the move $\sq^{N}$, the
environment (Opponent) provides a new fresh value $?Z$ for it.
The symbol $X$ is used to keep track of the current value of $x$.
Whenever a new value for $N$ is provided, the term has three possible options
depending on the current values of $Z$ and $X$: it can terminate successfully with $\sd$;
it can execute $\sa$ and terminate; or it can run the assignment $x \aasg x+1$ and ask
for a new value of $N$.

\begin{figure}
 \centerline{\scalebox{1.6}{\psfig{figure=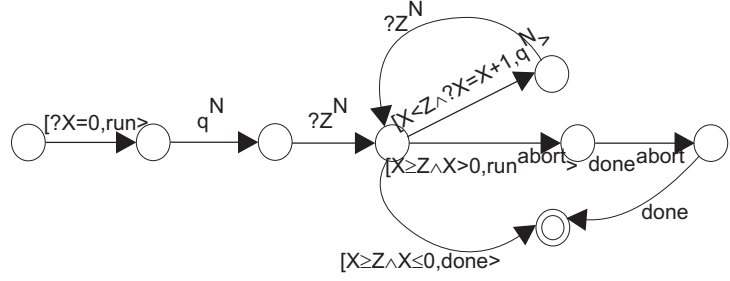}}}
 \caption{The strategy for $M_2$. }
\label{warm2}
\end{figure}

The shortest unsafe play found in the model is:
\[
\ibr{X=0,\sr} \ \sq^{N} \  {Z}^{N} \ \ibr{X \geq Z \land X>0,\sr^{abort}} \ \ \sd^{abort} \ \sd
\]
But the play condition for it, $X=0 \land X \geq Z \land X>0$, is inconsistent.
The next unsafe play is:
\[
\ibr{X_1=0,\sr} \ \sq^{N} \  {Z_1}^{N} \ \ibr{X_1<Z_1 \land X_2=X_1+1,\sq^{N}} \  {Z_2}^{N} \
  \ibr{X_2 \geq Z_2 \land X_2>0,\sr^{abort}} \ \ \sd^{abort} \ \sd
\]
Now Yices reports that the condition for this play is satisfiable,
yielding a possible assignment of concrete values to symbols that makes
the condition true: $X_1=0$, $Z_1=1$, $X_2=1$, $Z_2=0$.
So it is a genuine counter-example, such that one corresponding concrete
unsafe play is: $\sr \cdot \sq^N \cdot 1^N \cdot \sq^N \cdot 0^N \cdot \sr^{abort} \cdot \sd^{abort} \cdot \sd$.
This play corresponds to a
computation which runs the body of while exactly once.

Let us modify the $M_2$ term as follows
\[
\mathsf{new_{int}} \, x:=0 \ \mathsf{in} \  \mathsf{while} \, (x<N) \ \mathsf{do} \ x:=x+1; \,
 \mathsf{if} \, (x>k) \ \mathsf{then} \ \mathsf{abort}
\]
where $k>0$ is any positive integer.
The model for this modified term is the same as shown in Figure~\ref{warm2},
except that conditions associated with letters $\sr^{abort}$ (resp., $\sd$)
are $X \geq Z \land X > k$ (resp., $X \geq Z \land X \leq k$).
In this case the ($k+1$)-shortest unsafe plays in the model are found to
be inconsistent. The first consistent unsafe play corresponds to
executing the body of while ($k+1$)-times, and one possible concrete
representation of it (as generated by Yices) is:
\[
\sr \cdot \sq^N \cdot 1^N \cdot \sq^N \cdot 2^N \cdot \ldots \cdot \sq^N \cdot (k+1)^N \cdot \sq^N \cdot 0^N \cdot \sr^{abort} \cdot \sd^{abort} \cdot \sd
\]
$\Box$
\end{example}

\section{Extensions} \label{array}

We now extend the language with arrays of length $k>0$.
They can be handled in two ways.
Firstly, we can introduce arrays as syntactic sugar by using
existing term formers. An array $x[k]$ is represented as a
set of $k$ distinct variables $x[0]$, $x[1]$, $\ldots$, $x[k-1]$, such that
\[
\begin{tabular}{@{}l}
$x[E] \ \equiv$ \\
$\quad \mathsf{if} \: E=0 \: \mathsf{then} \: x[0] \: \mathsf{else} \: $ \\
$\quad \ \ldots$ \\
$\qquad \mathsf{if} \: E=k-1 \: \mathsf{then} \: x[k-1] \:
\mathsf{else} \: \mathsf{skip} \: (\mathsf{abort})$
\end{tabular}
\]
If we want to verify whether array out-of-bounds errors are present
in the term, i.e.\ there is an attempt to access elements out of the
bounds of an array, we execute $\mathsf{abort}$ instead of $\mathsf{skip}$ when $E \geq k$.
This approach for handling arrays
is taken by the standard representation of game semantics \cite{GM,DL}.

Secondly, since we work with symbols we can have more efficient representation
of arrays with unconstrained length.
While in the first approach the length of an array $k$ must be a concrete positive integer,
in the second approach $k$ can be represented by a symbol.
We use the support that Yices provides for arrays by enabling:
function definitions, function updates, and lambda expressions.
For each array $x[k]:\mathsf{var}D$, we can define a function symbol $X$ ($X: int \to D$)
 in Yices as:
\[
(define \, X::(\to int \, D))
\]
The function symbol $X$ can be initialized and updated as follows:
\[
\begin{array}{l}
(lambda \, (index::int) \, val  ) \\
(update \, X \, (index) \, val)
\end{array}
\]

A non-local array element is expressed as follows.
\[
\begin{array}{l}
\sbr{\Gamma, x[k] \vdash x[E] : \mathsf{var}D} = \sbr{\Gamma \vdash E : \mathsf{expint}^{\san{1}}} \comp_{\mathcal{A}_{\sbr{\mathsf{expint}}}^{gu \, \san{1}}} \sbr{\Gamma, x[k] \vdash x[-] : \mathsf{var}D} \\
\sbr{\Gamma, x[k] \vdash x[-] : \mathsf{var}D} = \srd \cdot \sq^{\san{1}} \cdot ?Z^{\san{1}} \cdot \ibr{Z<k,\srd^{\san{x[Z]}}} \cdot ?Z'^{\san{x[Z]}} \cdot Z' + \\
 \qquad \qquad \qquad \qquad \qquad \sw(?Z') \cdot \sq^{\san{1}} \cdot ?Z^{\san{1}} \cdot \ibr{Z<k,\sw(Z')^{\san{x[Z]}}} \cdot \sok^{\san{x[Z]}} \cdot \sok
\end{array}
\]
If we want to check for array out-of-bounds errors, we extend this interpretation by
including plays that perform moves associated with $\mathsf{abort}$ command when $Z \geq k$.
For example, the de-referencing (reading) part of the interpretation will be given as follows:
\[
\srd \cdot \sq^{\san{1}} \cdot ?Z^{\san{1}} \cdot \big( \ibr{Z<k,\srd^{\san{x[Z]}}} \cdot ?Z'^{\san{x[Z]}} \cdot Z'
\ + \  \ibr{Z \geq k,\sr^{\san{\sa}}} \cdot \sd^{\san{\sa}} \cdot 0  \big)
\]

The automaton $A$ for $\sbr{\Gamma \vdash \mathsf{new}_D \, x[k]:=v \, \mathsf{in} \, M}$,
where $A_M$ represents $\sbr{\Gamma, x[k] \vdash M}$,
is obtained as follows.
We first construct $A_{\varepsilon}$ by eliminating $x$-tagged moves from
$A_M$.
\[\begin{array}{l}
 Q_{\varepsilon}=Q_M \qquad i_{\varepsilon}=i_M \qquad F_{\varepsilon}=F_M \\
 \delta_{\varepsilon} \! = \! \{ i_M \stackrel{\ibr{X(j):=v \land b,m}}{\longrightarrow} q  \mid i_M \stackrel{\ibr{b,m}}{\longrightarrow} q \in \delta_M \} \ + \\
 \qquad \{ q_1 \stackrel{\ibr{b,m}}{\longrightarrow} q_2  \mid q_1 \stackrel{\ibr{b,m}}{\longrightarrow} q_2 \in \delta_M, m \nin \!\! \{ \sw(a)^{\san{x}}, \sok^{\san{x}}, \srd^{\san{x}}, a^{\san{x}} \! \} \}  \\
 \qquad \{ q_1 \stackrel{\ibr{X(a'):=a \land b_1 \land b_2,\varepsilon}}{\longrightarrow} q_2  \mid \exists q. ( q_1 \stackrel{\ibr{b_1,\sw(a)^{\san{x[a']}}}}{\longrightarrow} q \in \delta_M, q \stackrel{\ibr{b_2,\sok^{\san{x[a']}}}}{\longrightarrow} q_2 \in \delta_M )  \} \\
\qquad \{ q_1 \stackrel{\ibr{a=X(a') \land b_1 \land b_2,\varepsilon}}{\longrightarrow} q_2  \mid \exists q. ( q_1 \stackrel{\ibr{b_1,\srd^{\san{x[a']}}}}{\longrightarrow} q \in \delta_M, q \stackrel{\ibr{b_2,{a}^{\san{x[a']}}}}{\longrightarrow} q_2 \in \delta_M )  \} \\
\end{array}
\]
We use $X(j):=v$ to mean that the function symbol $X$ is initialized to $v$ for all its arguments,
while $X(a'):=a$ means that $X$ at argument $a'$ is updated to $a$.
The final automaton $A$ is generated by removing $\varepsilon$-letters from $A_{\varepsilon}$,
similarly as it was done for the case of $\mathsf{new}_D$ in Theorem~\ref{sym_aut}.

\section{Implementation} \label{app}

We have developed a prototype tool in Java, called \textsc{Symbolic GameChecker}, which
automatically converts an IA$_2$ term with integers into a symbolic automaton
which represents its game semantics. The model is then used to verify safety
of the term. Further examples as well as detailed reports of how they execute
 on \textsc{Symbolic GameChecker} are available from:\\
\verb|http://www.dcs.warwick.ac.uk/~aleks/symbolicgc.htm|.

Along with the tool we have also implemented in Java our own library of classes for
working with symbolic automata. We could not just reuse some of the existing libraries
for finite-state automata, due to the specific nature of symbolic automata we use.
The symbolic automata generated by the tool is checked for safety.
We use the breadth-first search algorithm to find the shortest unsafe play in the model.
Then the Yices is called to check consistency of its condition.
If the condition is found to be consistent, the unsafe play is reported; otherwise we
search for another unsafe play. If no unsafe play is discovered or all unsafe plays are
found to be inconsistent, then the term is deemed safe.
The tool also uses a simple forward reachability algorithm to remove all unreachable
states of a symbolic automaton.

Let us consider the following implementation of the linear search algorithm.
\[
\begin{array}{l}
 x[k] \, : \, \mathsf{varint}^{x[-]}, \ y \, : \, \mathsf{expint}^{y}, \ \mathsf{abort} \, : \, \mathsf{com}^{abort} \ \vdash \\
  \qquad \mathsf{new}_{int} \, i \aasg 0 \, \mathsf{in} \\
 \qquad \mathsf{new}_{int} \,  p \aasg y \, \mathsf{in} \\
 \qquad \mathsf{while} \, (i<k) \, \mathsf{do} \, \{ \\
 \qquad \quad \, \mathsf{if} \, (x[i] = p) \, \mathsf{then} \, \mathsf{abort}; \\
 \qquad \quad i := i + 1; \, \} \\
 \qquad : \mathsf{com}
\end{array}
\]

The program first remembers the input expression
$y$ into a local variable $p$.
The non-local array $x$ is then searched for an occurrence
of the value stored in $p$.
If the search succeeds, then $\mathsf{abort}$ is executed.

\begin{figure*}
\centerline{\scalebox{1.6}{\psfig{figure=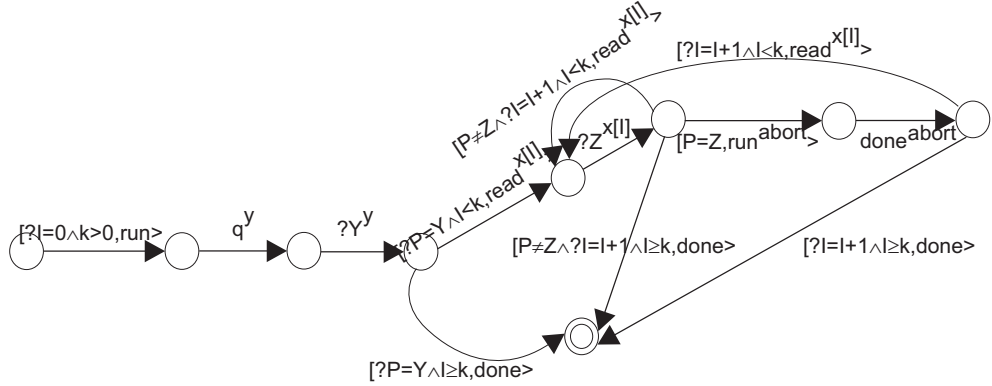}}}
\caption{The symbolic model for linear search. } \label{linear.model}
\end{figure*}

The symbolic model for this term is shown in Fig.~\ref{linear.model}, where for simplicity
array out-of-bounds errors are not taken in the consideration.
If the value read from the environment for $y$ has occurred in $x$, then an unsafe behaviour of the term
exists.
So this term is unsafe, and the following counter-example is found:
\[
\begin{array}{l}
\ibr{I_1=0 \land k>0, \sr} \ \ \sq^{y} \ \ Y^{y} \ \  \ibr{P=Y \land I_1<k,\srd^{x[I_1]}} \  \ Z^{x[I_1]} \\
 \qquad \ibr{Z=P ,\sr^{abort}} \ \ \sd^{abort} \ \ \ibr{I_2=I_1+1 \land I_2 \geq k, \sd}
\end{array}
\]
This play corresponds to a term with an array $x$ of size $k=1$, where
the values read from $x[0]$ and $y$ are equal.

Overall, the symbolic model for linear search term has 9 states and the total time
needed to generate the model and test its safety is less than 1 sec.
We can compare this approach with the tool in \cite{DL}, where the standard representation
based on CSP process algebra of terms with finite data types is used.
We performed experiments for the linear search term with different sizes of $k$ and
all integer types replaced by finite data types. The types of $x$, $y$, and $p$ is $int_n$,
i.e. they contain $n$ distinct values $\{0, \ldots n-1 \} $, and the type of the index $i$ is $int_{k+1}$,
i.e. one more than the size of the array.
Such term was converted into a CSP process \cite{DL}, and then the FDR model checker was used
to generate its model and test its safety.
Experimental results are shown in Table~\ref{results}, where
we list the execution time in seconds,
and the size of the final model in number of states.
The model and the time increase very fast as we increase the sizes of $k$ and $n$.
We ran FDR and \textsc{Symbolic GameChecker} on a Machine AMD
Phenom II X4 940 with 4GB RAM.

\begin{table}
\begin{center}
{\renewcommand{\arraystretch}{1}
\begin{tabular}{| r | c | c | c | c | c | c |} \hline
& \multicolumn{2}{| c |}{$n=2$} & \multicolumn{2}{| c |}{$n=3$}  \\ \cline{2-5} $k$ & Time & Model
 & Time & Model  \\
 \hline 1 & $<1$ & 11 & $<1$ & 13  \\
 \hline 5 & 1 & 43 & 1 & 61
\\ \hline 10 & 2 & 83 & 2 & 121
\\ \hline 15 & 5 & 123
& 6 & 181   \\
\hline
\end{tabular}}
\caption{Verification of the linear search with finite data
}\label{results}
\end{center}
\end{table}

\section{Conclusion} \label{concl}

We have shown how to reduce the verification of safety of game-semantics infinite-state
models of IA$_2$ terms to the checking of the more abstract finite symbolic automata.

Counter-example guided abstraction refinement procedures (ARP) \cite{DGL.SAS,DGL.SPIN}
can also be used for verification of terms with infinite integers.
However, they find solutions after performing a few iterations in order to adjust
integer identifiers to suitable abstractions. In each iteration, one abstract term
is checked. If an abstract term needs larger abstractions, then it is likely
to obtain a model with very large state space, which is difficult (infeasible)
to generate and check automatically.
The symbolic approach presented in this paper provides solutions in only one
iteration, by checking symbolic models which are significantly smaller than
 the abstract models in ARP. The possibility to handle arrays with
 unconstrained length is another important benefit of this approach.
 Extensions to nondeterministic \cite{D10} and concurrent \cite{GMO} terms can be
 interesting to consider.


\providecommand{\urlalt}[2]{\href{#1}{#2}}
\providecommand{\doi}[1]{doi:\urlalt{http://dx.doi.org/#1}{#1}}

\bibliographystyle{eptcs}

\end{document}